%% file: scheduling_caching_journal_single_column.tex
\documentclass[journal,12pt,draftclsnofoot,onecolumn,letterpaper]{IEEEtran}
\usepackage{bbm,dsfont,mathrsfs,fixmath,amsthm}
\usepackage{amsmath,amssymb,eucal,graphicx}
\usepackage{stmaryrd,cite}
\usepackage{amsmath,amstext,amsfonts,amssymb}
\usepackage{epsfig}
\usepackage{exscale}
\usepackage{enumerate}
\usepackage{mathtools}
\usepackage{multirow}

\usepackage[geometry]{ifsym}
\usepackage{ifsym}

\usepackage{amscd}
\usepackage{dsfont}
\usepackage{graphics}
\usepackage{psfrag}
\usepackage{rotating}
\usepackage{amsfonts}
\usepackage{url}
\usepackage{color}
\usepackage[caption=false,font=footnotesize]{subfig}
\usepackage{epstopdf}
\usepackage{amsthm}
\usepackage{pgfplots}
\usepackage{tikz}
\usepackage{makecell}

\newtheorem{cor}{Corollary}
\newtheorem{proposition}{Proposition}
\newtheorem{property}{Property}
\newtheorem{theorem}{Theorem}

\newcommand{\Rbl}{R_{\rm sum,bl}}	

\newcommand{\bp}{\noindent{\textbf{Proof.}}\ }
\newcommand{\ep}{\hfill $\Box$}

\newcommand{\cond}{\,\vert\,}

\newcommand{\indic}{\mathbf{1}}

\newfont{\bbb}{msbm10 scaled 500}

\newfont{\bb}{msbm10 scaled 1100}
\newcommand{\CC}{\mbox{\bb C}}
\newcommand{\RR}{\mbox{\bb R}}

\newcommand{\EE}{\mbox{\bb E}}

\newcommand{\PP}{\mbox{\bb P}}

\newcommand{\ev}{{\bf e}}

\newcommand{\hv}{{\bf h}}

\newcommand{\rv}{{\bf r}}

\newcommand{\wv}{{\bf w}}

\newcommand{\xv}{{\bf x}}
\newcommand{\yv}{{\bf y}}

\newcommand{\PtoK}{ \overset{\PP}{\underset{K \to \infty}{\to}} }
\newcommand{\astoK}{ \overset{a.s.}{\underset{K \to \infty}{\to}} }
\newcommand{\toK}{ \overset{}{\underset{K \to \infty}{\to}} }

\newcommand{\Id}{{\bf I}}

\newcommand{\Rm}{{\bf R}}

\newcommand{\Jc}{{\cal J}}
\newcommand{\Kc}{{\cal K}}

\newcommand{\Nc}{{\cal N}}

\newcommand{\betav}{\hbox{\boldmath$\beta$}}

\renewcommand{\arg}{{\hbox{arg}}}
\newcommand{\var}{{\hbox{var}}}

\DeclareFontFamily{U}{cmfi}{}
\DeclareFontShape{U}{cmfi}{m}{n}{ <-> cmfi10 }{}
\DeclareSymbolFont{CMFI}{U}{cmfi}{m}{n}

\def\argmax{\mathop{\rm argmax}}
\def\argmin{\mathop{\rm argmin}}

\newcommand{\vh}{\pmb{h}}

\newcommand{\vu}{\pmb{u}}

\renewcommand{\Rm}{\pmb{R}}

\renewcommand{\ev}{\pmb{e}}

\renewcommand{\hv}{\pmb{h}}

\renewcommand{\rv}{\pmb{r}}

\renewcommand{\wv}{\pmb{w}}
\renewcommand{\xv}{\pmb{x}}
\renewcommand{\yv}{\pmb{y}}

\begin{document}

\title{Utility Optimal Scheduling for Coded Caching in General Topologies}

\author{\IEEEauthorblockN{}
\IEEEauthorblockA{Richard Combes, Asma Ghorbel, Mari Kobayashi, and Sheng Yang \\
\thanks{This work was partly supported by Huawei Technologies France.}
LSS, CentraleSup\'elec \\
Gif-sur-Yvette, France\\
 {\tt \{firstname.lastname\}@centralesupelec.fr}
}
}

\maketitle
\begin{abstract}

We consider {\it coded caching} over the fading broadcast channel, where the users, equipped with a memory of finite size, experience
asymmetric fading statistics. It is known that a naive application of coded caching over the channel at hand performs poorly especially in the
regime of a large number of users due to a vanishing multicast rate. We overcome this detrimental effect by a careful design of opportunistic
scheduling policies such that some utility function of the long-term average rates should be maximized while balancing fairness among users.  
In particular, we propose a threshold-based scheduling that requires only statistical channel state information and one-bit feedback from each
user. More specifically, each user indicates via feedback whenever its SNR is above a threshold determined solely by the fading statistics and
the fairness requirement. Surprisingly, we prove that this simple scheme achieves the optimal utility in the regime of a large number of users.
Numerical examples show that our proposed scheme performs closely to the scheduling with full channel state information, but at a significantly reduced complexity.
\end{abstract}
\section{Introduction}\label{section:introduction}
\input{Introduction2}

\section{System Model}\label{section:model}
\input{SystemModel}

\section{Problem Formulation}\label{section:formulation}
\input{Formulation}

\section{Fair scheduling}\label{sec:fair_scheduling}
\input{FairScheduling}

\section{Fair scheduling for a large number of users}\label{sec:schduling_analysis}
\input{SchedulingAnalysis}

\section{Numerical Experiments}\label{sec:numericalEx}
\input{Numerical-Examples}
\begin{figure}[!h]
  \begin{minipage}[t]{0.5\linewidth}
\centering
\includegraphics[width=0.9\textwidth,clip=]{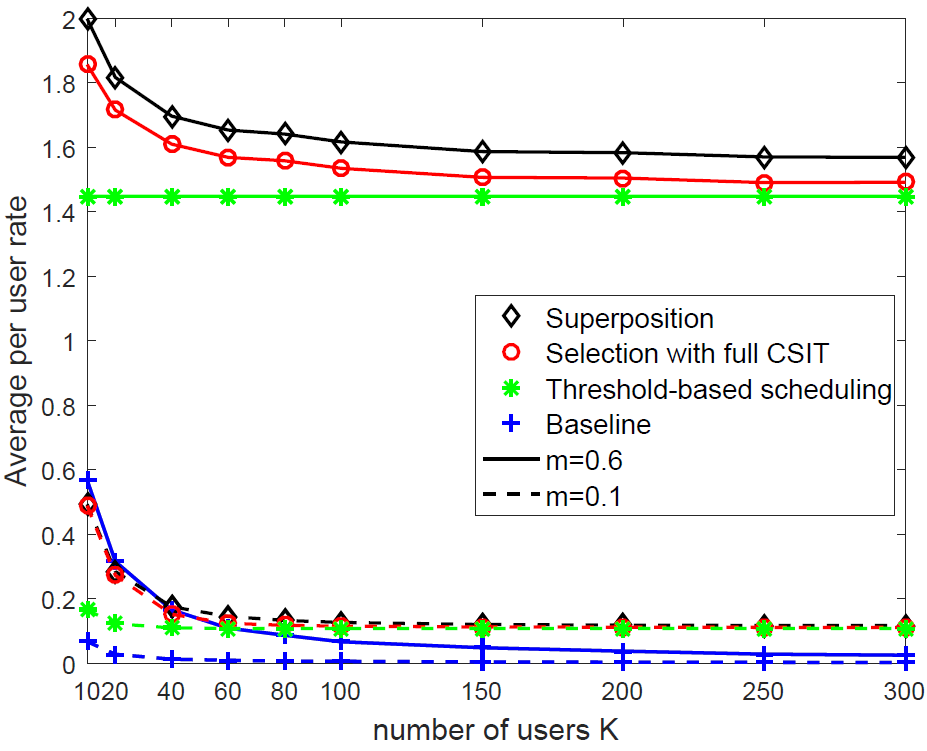}
\caption{Average per user rate vs $K$ for $\alpha=0$, $P=10$dB and $m=[0.1,0.6]$.}
\label{fig:ka0}

\end{minipage}
  \begin{minipage}[t]{0.5\linewidth}
\centering
\includegraphics[width=0.9\textwidth,clip=]{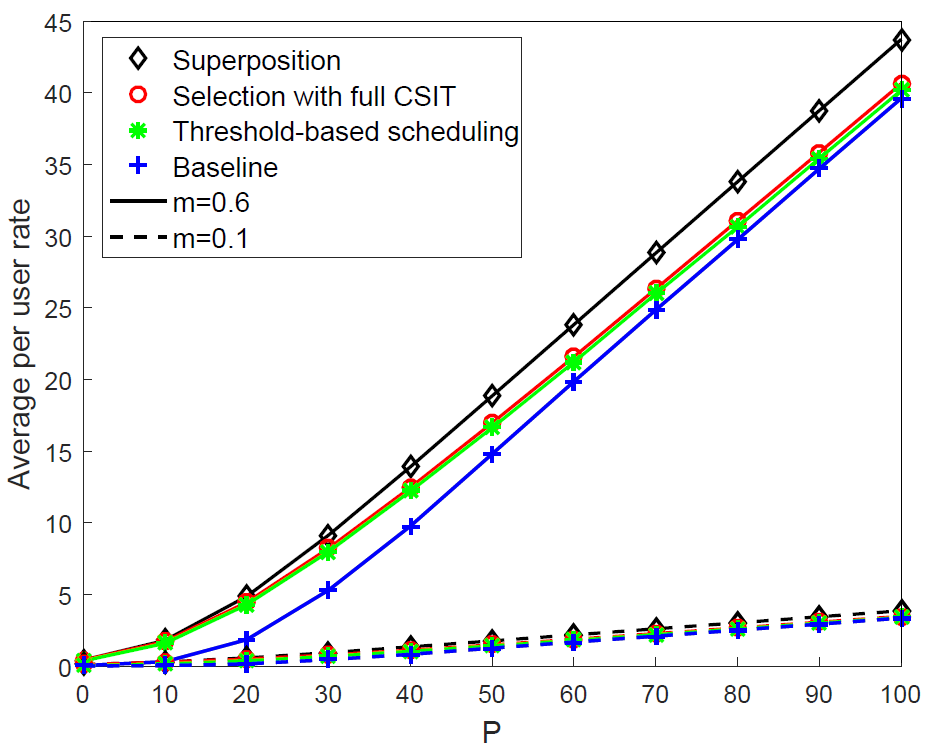}
\caption{Average per user rate vs $P$ for $\alpha=0$, $K=20$ and $m=[0.1,0.6]$.}
\label{fig:pa0}

\end{minipage}
\end{figure}

\begin{figure*}
  \begin{minipage}{0.5\linewidth}
\begin{center}
\includegraphics[width=0.9\textwidth,clip=]{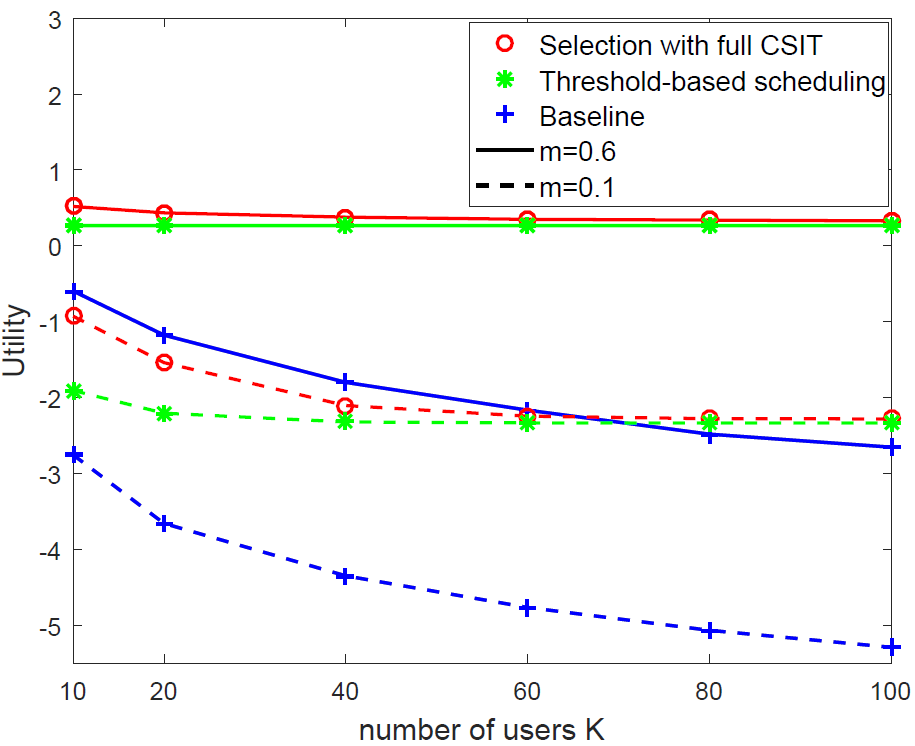}
\caption{Utility vs $K$ for $\alpha=1$, $P=10$dB and $m=[0.1,0.6]$.}
\label{fig:ka1}
\end{center}
\end{minipage}
\begin{minipage}{0.5\linewidth}
\begin{center}
\includegraphics[width=0.9\textwidth,clip=]{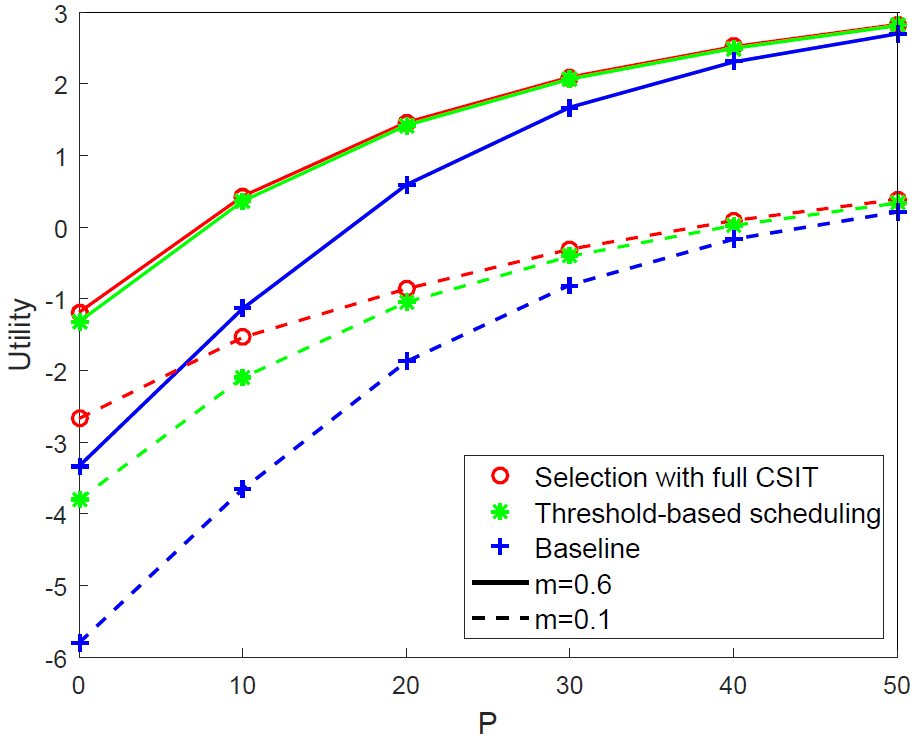}
\caption{Utility vs $P$ for $\alpha=1$, $K=20$ and $m=[0.1,0.6]$.}
\label{fig:pa1}
\end{center}
\end{minipage}
\end{figure*}

\begin{figure*}
  \begin{minipage}{0.5\linewidth}
\begin{center}
\includegraphics[width=0.9\textwidth,clip=]{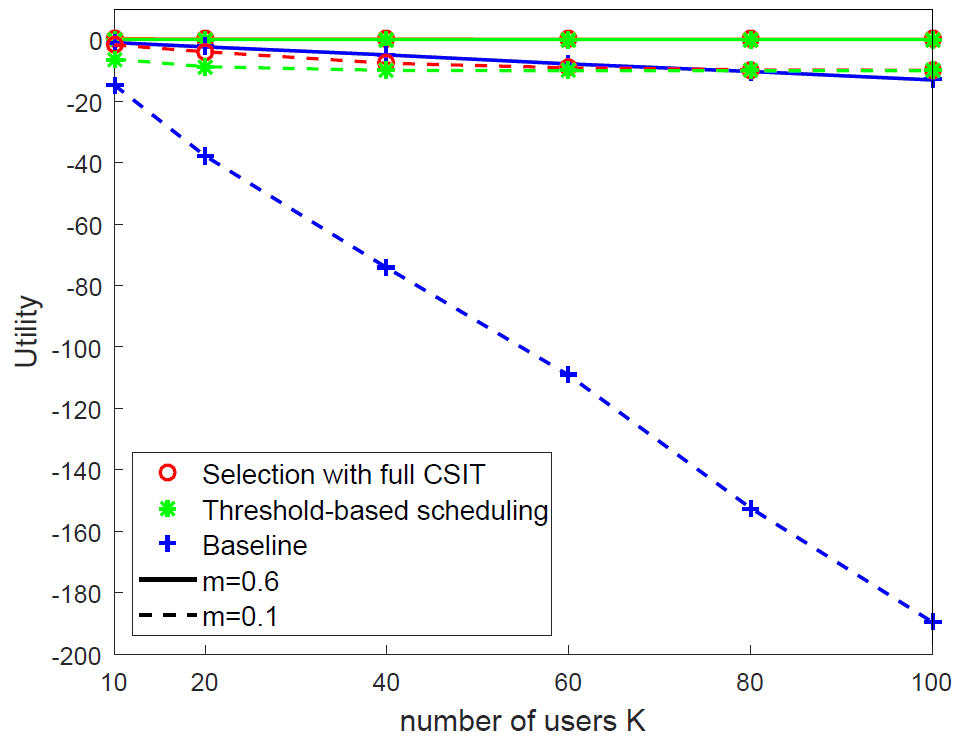}
\caption{Utility vs $K$ for $\alpha=2$, $P=10$dB and $m=[0.1,0.6]$.}
\label{fig:ka2}
\end{center}
\end{minipage}
\begin{minipage}{0.5\linewidth}
\begin{center}
\includegraphics[width=0.9\textwidth,clip=]{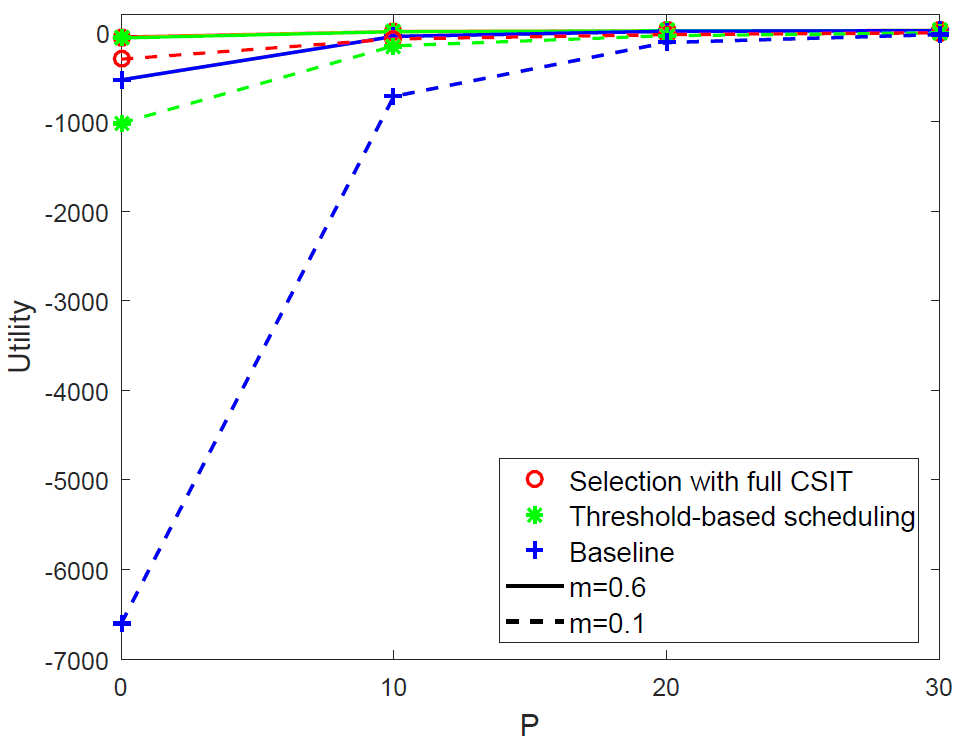}
\caption{Utility vs $P$ for $\alpha=2$, $K=20$ and $m=[0.1,0.6]$.}
\label{fig:pa2}
\end{center}
\end{minipage}
\end{figure*}


\section{Conclusion} 

Recent works have revealed that the theoretical gain of coded caching is sensitive to the behavior of the multicast rate of the
underlying channel and might vanish in the regime of a large number of users. In order to overcome such detrimental effect, we have
studied opportunistic scheduling schemes for coded caching over the asymmetric fading broadcast channel. For the alpha-fairness
utility function of the long-term average rates, we have proposed a simple threshold-based scheduling policy, which requires only
statistical channel knowledge and can be implemented by a simple one-bit feedback from each user. Our striking result, through
rigorous and rather involved analysis, demonstrates that such threshold-based policy is asymptotically optimal as the number of users
grows. Additionally, the numerical examples show that our proposed policy incurs a negligible loss with respect to the optimal
scheduling scheme (requiring full channel knowledge) for a reasonable number of users, i.e., between 20 to 100 users depending on the fairness parameter and the memory size. 
\bibliographystyle{IEEEtran}
\bibliography{caching}

\input{Appendix}

%
\end{document}

%% file: Introduction2.tex
Content delivery applications such as video streaming are envisioned to represent nearly 75\% of the mobile data traffic by 2020 \cite{index2015global}. 
The skewness of the video traffic together with the ever-growing cheap on-board storage memory suggests
that the quality of experience can be improved by caching popular content close to the end-users in wireless networks. 
Recent works have studied the gains provided by caching under various models and assumptions (see e.g. \cite{maddah2013fundamental,golrezaei2011femtocaching} and references therein). 
In this work, we consider content delivery using {\it coded caching} \cite{maddah2013fundamental} in a wireless network where a server is connected to $K$ users each equipped with a cache of finite memory. By a careful design of sub-packetization and cache placement, it is possible to create a multicast signal simultaneously useful for many users and thus decrease the delivery time. More specifically,  it has been proved that 
 the delivery time to satisfy $K$ distinct requests converges to a constant in the regime of a large number $K$ of users. In other words, the sum content delivery rate, defined as the total amount of requested bits divided by the delivery time, grows linearly with $K$. This striking result has motivated a number of follow-up works in order to study coded caching in more realistic scenarios (see e.g. \cite[Section VIII]{maddah2013fundamental}). 


Albeit conceptually and theoretically appealing, the promised gain of coded caching relies on some unrealistic assumptions~(see. e.g.
\cite{misconceptions}). In particular, \cite{submittedTWC2017, ghorbel2017opportunistic, SubmittedWiopt2016} have revealed that the
scalability of coded caching is very sensitive to the behavior of the multicast rate supported by the bottleneck link. It is worth
recalling that the multicast capacity of the fading broadcast channel is limited by the channel quality of the weak users, i.e., the
users whose channel gain is the smallest.   
Focusing on the case of the i.i.d. quasi-static Rayleigh fading channel, the works ~\cite{NgoAllerton2016,submittedTWC2017} further
showed that the long-term sum content delivery rate does not grow with the system dimension if coded caching is naively applied to this channel. In fact, the long-term average multicast rate of the i.i.d. Rayleigh fading channel vanishes, as it scales as ${\cal O}({1 \over K})$ as $K \to \infty$, \cite{jindal2006capacity}. 
When the users experience asymmetric fading statistics, the long-term average multicast rate is essentially limited by the users with
poor channel statistics. Therefore, the performance of coded caching may degrade even further, since nearly the whole resource is
wasted to enable the weak users to decode the common message. These observations have inspired a number of recent works to overcome these drawbacks  \cite{ghorbel2017opportunistic, SubmittedWiopt2016,   Shariatpanahi2016multiserver, NgoAllerton2016, zhang2017fundamental, submittedTWC2017, Shariatpanahi2016multiantenna, shariatpanahi2017physical}. 
The works \cite{Shariatpanahi2016multiserver,NgoAllerton2016,zhang2017fundamental,submittedTWC2017,Shariatpanahi2016multiantenna}
have considered the use of multiple antennas, while \cite{bidokhti2016noisy,zhang2016wireless} have proposed several interference
management techniques. Other recent works have studied opportunistic scheduling \cite{submittedTWC2017, ghorbel2017opportunistic,
SubmittedWiopt2016} in this context.  Finally, the interplay between the fairness and the gain of  coded caching has been studied in
a recent work \cite{SubmittedWiopt2016}. Although both the current work and  \cite{SubmittedWiopt2016} consider the same channel
model and address a similar question, they differ in their objectives and approaches. In \cite{SubmittedWiopt2016}, a new queueing
structure has been proposed to deal jointly with admission control, routing, as well as scheduling for a finite number of users. The
performance analysis built on the Lyapunov theory. The current work highlights the scheduling part and provides a rigorous analysis on the long-term average per-user rate in the regime of a large number of users. 

As a non-trivial extension of \cite{submittedTWC2017, ghorbel2017opportunistic}, we study opportunistic scheduling in order to achieve a scalable
sum content delivery while ensuring some fairness among users. To capture these two contrasted measures, we formulate our objective function by
an alpha-fairness family of concave utility functions~\cite{mowalrand}. Our main contributions of this work are three-fold: 
\begin{enumerate}
\item We propose a simple threshold-based scheduling policy and determine the threshold as a function of the fading statistics for each fairness parameter $\alpha$. 
Such threshold-based scheme exhibits two interesting features. On the one hand, the complexity is linear in $K$ and significantly reduced  with
respect to the original problem where the search is done over $K^2$ variables. On the other hand, a threshold-based policy does not require the
exact channel state information but only a one-bit feedback from each user. Namely, each user indicates whether its measured SNR is above the
threshold set before the communication. 
 A special case of the symmetric fading and the sum rate objective ($\alpha=0)$, our proposed scheme boils down to the scheme in \cite{submittedTWC2017, ghorbel2017opportunistic}. 
\item We prove that the proposed threshold-based scheduling policy is asymptotically optimal in Theorem 3. Namely, the utility achieved by our
  proposed policy converges to the optimal value as the number of users grows. The proof of Theorem 3 involves essentially three steps. First, we characterize the lower and upper bounds on the long-term average rate of each user. Second, we prove that the size of the selected user set grows unbounded as the number of users grows. Finally, we prove the convergence of the utility value. 
\item Our numerical experiments show that the proposed scheme indeed achieves a near-optimal performance. Namely, it converges to the
  selection scheme with full channel knowledge as the number of users and/or SNR increases. Such  scheme is therefore appropriate for
  a large number of users. In addition, the multicast rate is less sensitive to the user in the worst fading condition in the large
  SNR regime. Furthermore, the speed of convergence increases with the memory size and/or $\alpha$-fair parameter. In fact, Property \ref{pt2} in subsection \ref{subsection:Acc} justifies the impact of the memory. 
\end{enumerate}

The remainder of the paper is organized as follows. Section~\ref{section:model} provides the system model and Section~\ref{section:formulation}
formulates the fair scheduling problem as maximizing the $\alpha$-fair utility. In Section~\ref{sec:fair_scheduling}, we define the optimal
policy as well as a class of threshold-based policies with reduced complexity. In Section~\ref{sec:schduling_analysis}, we state and prove the main result,  that is, the threshold-based scheduling
policy achieves the optimal utility in the regime of a large number of users. We further characterize the threshold-based policy for different fairness criteria.  Section \ref{sec:numericalEx} provides numerical examples to validate our analysis in previous sections and compare the performance of the proposed threshold-based policy with other schemes.

Throughout the paper, we use $[k]$ to denote the set of integers $\{1, \dots, k\}$, and $f(x)\sim g(x)$ means that $\lim\limits_{x\to\infty}\frac{f(x)}{g(x)}=1$. We use the notation $\overset{\PP}{\underset{}{\to}}$ to denote convergence in probability and $\overset{a.s.}{\underset{}{\to}}$ to denote almost sure convergence.


%% file: SystemModel.tex
\begin{figure}[t]
	\centering
	\includegraphics[width=0.6\textwidth,clip=]{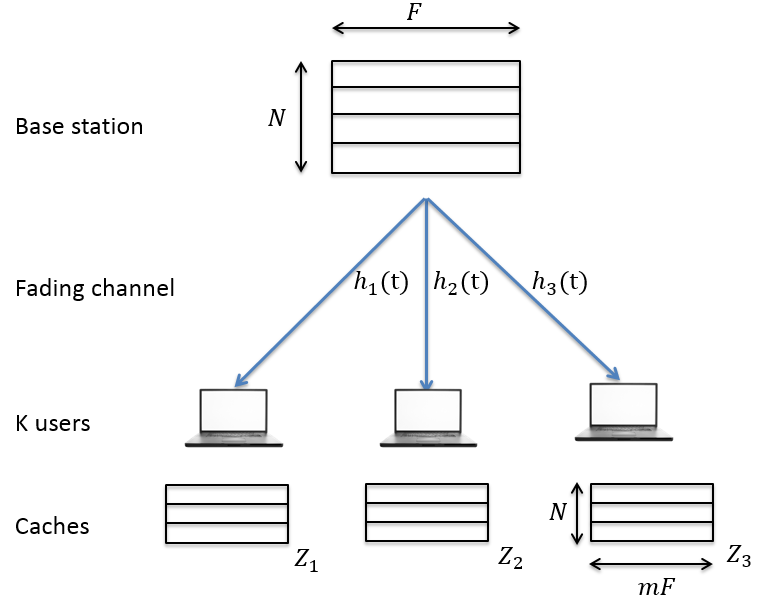}
	\vspace{-1em}
	\caption{System model with $K=3$.}
	\vspace{-1em}
	\label{fig:sysm}
\end{figure} 

We consider a content delivery system where a  server with $N$ files wishes to convey the requested files to $K$ users over a
wireless downlink channel. We assume that $N$ files are of equal size of $F$ bits and have equal popularity, while each user has a
cache of size $M F$ bits, where $M \ge 1$ denotes the cache size measured in files. We often use the normalized cache size denoted by
$m=M/N$. In this work, we focus on the regime of a large number of files, i.e., $N\geq K$, and assume that the requests from the users are all {\it distinct}. 
Further, each user can prefetch some content to fill their caches during off-peak hours, prior to the actual request.
We consider mainly the decentralized caching scheme of \cite{maddah2013decentralized}, where each user independently caches a subset of $mF$ bits of file $i$, chosen uniformly at random for $i=1,\dots, N$ under the memory constraint of $MF$ bits. By letting $W_{i|\Jc}$ denote the sub-file of $W_i$ stored exclusively in the cache memories of the user set $\Jc$, the cache memory $Z_k$ of user $k$ after decentralized caching is given by
\begin{align} \label{eq:Zk}
Z_k =\{ W_{i \cond \Jc}:  \;\; \forall \Jc \subseteq[K], \forall \Jc \ni k , \forall i \in [N] \}.
\end{align}
 Once the user requests are revealed, the server generates and sequentially conveys the codewords intended to each subset of users. 
Namely, assuming that user $k$ requests file $W_k$ for all $k$, the codeword intended to the subset $\Jc$ is given by
\begin{align}\label{eq:V}
V_{\Jc}=\oplus_{k\in \Jc}W_{k|\Jc\setminus\{k\}},
\end{align}
where $\oplus$ denotes the bit-wise XOR operation. The main idea here is to create a codeword useful to a subset of users by exploiting the receiver side information established during the placement phase. 
It has been shown for decentralized caching in \cite{maddah2013decentralized} that the delivery time, or the number of multicast transmissions, needed to satisfy $K$ {\it distinct} demands over the error-free shared link is 
\begin{align} \label{eq:defT}
T(m,K)= \left(1-m\right) 
\frac{1-(1-m)^K}{m}.
\end{align}

In order to ensure reliable delivery in a wireless channel, the codewords described in~\eqref{eq:V} should then be encoded with a proper
channel code in the physical layer. In this work, the physical layer is modeled as a single-antenna quasi-static fading Gaussian broadcast
channel. Specifically, we assume that the channel state remains constant during the transmission of any channel codeword, or, equivalently, any
physical layer frame. 
Let us focus on the transmission~$t$, where $t$ can be considered as the frame index. 
For a given channel codeword $\xv(t)\in\CC^n$, user $k$ receives
\begin{align}\label{eq:FadingBC}
\yv_k(t) = \sqrt{\tilde{h}_k}(t) \, \xv(t) + \wv_k(t),
\end{align}
where the input satisfies the power constraint $\|\xv(t)\|^2\leq n P$; $\{\tilde{h}_k(t)\}_k$ are the fading gains independently distributed
over users\footnote{Note that the phase of the channel coefficient is ignored since each receiver can rotate the signal to remove the phase. };
$\wv_k(t)\sim \Nc_{\CC} (0, \Id_n)$ is the additive white Gaussian noise assumed to be independent and identically distributed across time and
users. For simplicity, we define $h_k = {P} \tilde{h}_k$ assumed to be exponentially distributed with mean $\gamma_k$. We assume that each user
$k$ knows its channel realization $h_k$. In addition, we are particularly interested in the long-term behavior~(e.g., time span of hours or
days) of the system. To simplify such analysis, we further assume that the channel coefficient of each user changes to an independent
realization from codeword to codeword according to the same distribution, i.e., $h_k(t)$ is i.i.d.~over~$t$ for a given $k$.

It is well-known that the multicast capacity of the channel at hand, or the common message rate,  is given by 
\begin{align} \label{eq:multicast_rate}
R_{\rm mc} (\hv) = \log\left(1+  \min_{j\in [K]} h_j\right)
\end{align}
and is limited by the user in the worst channel condition. It has been proved in \cite{NgoAllerton2016} that such limitation is detrimental for a scalable content delivery network. 
To see this, let us first define the sum content delivery rate when coded caching is applied directly to the fading broadcast channel. 
In order to satisfy the distinct demands from $K$ users, or to \emph{complete} in total $K F$ demanded bits, 
we need to send $T(m,K) F$ bits over the wireless link. The corresponding transmission takes $\frac{T(m,K) F}{ R_{\rm mc}(\hv)} $
units of time. As a result, the sum content delivery rate of a naive application of coded caching 
for a given channel realization $\hv$ is given by 
\begin{align} \label{eq:Rmul}
 \frac{K}{T(m,K)} R_{\rm mc}(\hv)
\end{align}%
measured in [bits/second/Hz]. For convenience, we call such a naive application as the ``baseline''~(``bl'') scheme where the base station
serves all $K$ users with the multicast rate limited by the worst user as in \eqref{eq:multicast_rate}. The corresponding (long-term)~average sum content delivery rate is given by 
\begin{align} \label{eq:baseline_rate}
\Rbl (K)=\frac{K}{T(m,K)} \EE[R_{\rm mc}(\hv)].
\end{align}
To gain an insight into the harmful effect, let us consider the case of symmetric fading statistics ($\gamma_k=\gamma, \forall k$). The average multicast capacity $\EE[R_{\rm mc} (\hv)]$ vanishes as ${\cal O}(1/K)$ for $K \to \infty$ \cite{jindal2006capacity}, the average sum content delivery rate converges to a constant, yielding a non-scalable system. 
More precisely, we recall the following result. 
\begin{proposition} 
  \label{BaselineScheme}
The long-term average sum content delivery rate of baseline scheme is given by
\begin{align}
\Rbl (K) = \frac{K}{T(m,K)} e^{{K \over \gamma}} E_1\left({K \over \gamma}\right),
\end{align}
where we define the exponential integral function $E_1(x)=\int_{1}^{+\infty} {e^{-xt}  \over t} dt$.
As $K \to \infty$, we have 
\begin{align}
	\Rbl(K) \sim {\gamma m \over  1-m}. 
\end{align}%
\end{proposition}
\begin{proof}
  Refer to appendix \ref{appendix:prop1}.
\end{proof}
This negative result motivates us to study some opportunistic scheduling strategy which benefits both from the coded caching gain and the
diversity of the underlying wireless channel, while ensuring certain fairness among users.  


%% file: Formulation.tex
In this section, we first review the fading Gaussian broadcast channel where the transmitter wishes to convey $2^K-1$ mutually independent
messages, each intended to a subset of users. We recall the capacity region achieved by superposition encoding and provide the optimal power
allocation. This serves as the ultimate upper bound on the fair scheduling problem. Then, we formulate our objective function by an alpha-fair
family of concave utility functions. 

\subsection{Capacity region of the fading Gaussian broadcast channel}
In this subsection, we review Theorems 1 and 2 of \cite{SubmittedWiopt2016,ghorbel2017opportunistic} which serve as the 
upper bound of more practical scheduling policies considered shortly.  
It readily follows that the channel in \eqref{eq:FadingBC} for a given channel realization $\hv$ is a degraded Gaussian broadcast channel. Without loss of generality, let us assume $h_1\geq \dots \geq h_K$. 
 Let us consider that the transmitter wishes to convey 
$2^K-1$ mutually independent messages, denoted by $\{M_{\Jc}\}$, where $M_{\Jc}$ denotes the message intended to the users in subset
$\Jc\subseteq [K]$. Each user $k$ must decode all messages $\{M_{\Jc}\}$ for $\Jc\ni k$. By letting $R_{\Jc}$ denote the multicast rate of the
message $M_{\Jc}$, we say that the rate-tuple $\Rm\in \RR_+^{2^K-1}$ is achievable if there exists some encoding and decoding functions such
that decoding error probability can be arbitrarily small with large codeword length~$n$. 
The capacity region is defined as the set of all achievable rate-tuples and is given by the following theorem. 
\begin{theorem}
  \label{theorem:Region} 
The capacity region $\Gamma(\hv)$ of a $K$-user degraded Gaussian broadcast channel with fading gains $h_1 \geq \dots \geq h_K$ and $2^K-1$ independent messages $\{M_{\Jc}\}$ is given by 
\begin{align}\label{eq:9}
\sum_{\Kc: k\in \Kc \subseteq [k]} R_{\Kc} & \leq \log\frac{1+ h_k \sum_{j=1}^{k} \beta_j }{1+ h_k\sum_{j=1}^{k-1} \beta_j }, \quad k=2,
\dots, K,
\end{align}
for non-negative variables $\{\beta_k\}$ such that $\sum_{k=1}^K \beta_k \leq 1$. 
\end{theorem}
\begin{proof}
The proof is quite straightforward and is based on rate-splitting and the private-message region of degraded broadcast channel. For
completeness, see details in Appendix \ref{subsection:th1}. 
\end{proof}
In order to characterize the boundary of the capacity region $\Gamma(\hv)$, we consider the weighted sum rate maximization given as
\begin{align}\label{eq:generalWSR}
\max_{\rv \in \Gamma(\hv)}\sum_{\Jc: \Jc \subseteq [K]} \theta_{\Jc} r_{\Jc}.
\end{align}
By exploiting a simple property of the capacity region, the problem at hand can be cast into a simpler problem as summarized below. 

\begin{theorem}
  \label{theorem:WSR}
The weighted sum rate maximization with $2^K-1$ variables in \eqref{eq:generalWSR} reduces to a simpler problem with $K$ variables, given by 
\begin{align}\label{eq:powerallocation}
f(\betav) = \sum_{k=1}^K \tilde{\theta}_k \log\frac{1+h_k \sum_{j=1}^{k} \beta_j }{1+ h_k\sum_{j=1}^{k-1} \beta_j }, 
\end{align}
where $\tilde{\theta}_k$ denotes the largest weight for user $k$ 
\begin{align}
\tilde{\theta}_k=\max_{\Kc: k\in \Kc \subseteq [k]}\theta_{\Kc}.
\end{align}%
\end{theorem}
\begin{proof}
  Refer to Appendix \ref{subsection:th2}.
\end{proof}

\subsection{Application to coded caching }\label{subsection:Acc}
By performing coded caching to the user subset $\Jc$, the total number of bits to be multicast to satisfy $|\Jc|$ distinct demands is equal to $T(m,|\Jc|)F$. 
By letting $R_{\Jc}$ denote the multicast rate of the codewords intended to user subset $\Jc$, the per-user rate after applying coded caching to subset $\Jc$ is given by $\frac{1}{T(m,|\Jc|)}R_{\Jc}$ for any user in $\Jc$. By simultaneously applying coded caching over different subset of users, the per-user data rate is given by
  \begin{align}\label{eq:rel}
U_k&=\sum_{\Jc:k\in\Jc\subseteq[K]}\frac{1}{T(m,|\Jc|)}R_{\Jc}.
\end{align}
Using \eqref{eq:rel}, the weighted sum of the individual user rates $\sum_{k=1}^{K}\tau_k U_k$, for any $\{\tau_k\}_{1\leq k\leq K}$, can be rewritten as: 
\begin{align}
\sum_{k=1}^{K}\tau_k U_k&=\sum_{k=1}^{K}\tau_k\sum_{\Jc:k\in\Jc\subseteq[K]}\frac{1}{T(m,|\Jc|)}R_{\Jc}\\
&=\sum_{\Jc:\Jc\subseteq[K]}\sum_{k:k\in\Jc}\frac{\tau_k}{T(m,|\Jc|)}R_{\Jc}\\
&=\sum_{\Jc:\Jc\subseteq[K]}\theta_{\Jc}R_{\Jc},
\end{align}
where $\theta_{\Jc}=\frac{\sum_{k:k\in\Jc}\tau_k}{T(m,|\Jc|)}$. Hence the problem can be reduced to that of the previous subsection. Throughout the paper we use three facts concerning the mapping $T$ which are stated below.

\begin{property}\label{pt2}
$T(m,k)$ converges to $T(m,\infty)\triangleq\frac{1-m}{m}$ when $k\rightarrow \infty$. The larger $m$ is, the faster it converges.
\end{property}
\begin{property}\label{pt1}
$T(m,k)$ is an increasing function of $k$ and so $T(m,1)\leq T(m,k)\leq T(m,\infty)$. 
\end{property}
\begin{property}\label{pt3}
$\frac{k}{T(m,k)}$ is an increasing function of $k$. 
\end{property}

\subsection{Objectives}

 Since implementing superposition encoding is complex, we now restrain ourselves to practical schemes which, for each channel realization $\hv$, select a group of users $\Jc \subset \{1,...,K\}$ to perform the delivery scheme of \cite{maddah2013decentralized} to $\Jc$ at rate $\log(1+\min_{j \in \Jc} h_j)$. Consider $\pi$ a scheduling policy, which is a mapping from $(\mathbb{R}^{+})^K$ to the set of subsets of $\{1,...,K\}$. For a channel
 realization $\hv = (h_1,\dots,h_K)$, the policy $\pi$ chooses a group of users $\Jc^\pi(\hv) \subset \{1,\dots,K\}$ for transmission, where
 the transmission strategy is the one described in previous sections. We denote by $\Pi$ the set of admissible policies. Given the policy $\pi$
 and channel realization $\hv$, user $i$ is served at the rate given by 
 \begin{align}\label{eq:90}
	{\indic\{i \in \Jc^\pi(\hv) \} \over T(m,|\Jc^\pi(\hv)|)} \log\left(1 + \min_{j \in \Jc^\pi(\hv)} h_j\right),
\end{align}
so that the rate depends on both the size of the selected group $|\Jc^\pi(\hv)|$ and the minimal channel gain $\min_{j \in \Jc^\pi(\hv)}
h_j$ among the chosen users. It is noted that for a fixed value of $\min_{j \in \Jc^\pi(\hv)} h_j$, the rate \eqref{eq:90} is a decreasing function of
the group size $|\Jc^\pi(\hv)|$ due to Property \ref{pt1}, while for a fixed group size, the rate \eqref{eq:90} is an increasing function of $\min_{j \in \Jc^\pi(\hv)} h_j$. 

Under policy $\pi$, the long-term average rate of user $i$ is the expectation of the instantaneous data rate over the channel realizations $\vh$:
$$
	U_i^\pi = \EE\left( {\indic\{i \in \Jc^\pi(\hv) \} \over T(m,|\Jc^\pi(\hv)|)} \log(1 + \min_{j \in \Jc^\pi(\hv)} h_j)\right).
$$ 
We are interested in utility-optimal scheduling, where the goal is to maximize some utility function of the long-term rates. We restrict
our attention to $\alpha$-fair allocations \cite{mowalrand}, namely,
\begin{align}\label{eq:utilitymax}
		\pi^\star \in \underset{\pi \in \Pi}{\argmax} \left \{{1 \over K} \sum_{i=1}^K g_{\alpha}( U_i^\pi)\right\} 
\end{align}
with 
$$	
	g_\alpha(x) = \begin{cases} 	
{x^{1-\alpha} - 1 \over 1 - \alpha}, &\text{ if } \alpha \ne 1, \\
\log(x), &\text{ if } \alpha = 1.
												\end{cases}
$$
It is noted that $\alpha \mapsto g_\alpha(x)$ is continuous for any fixed $x$ since $\lim_{\alpha \to 1} {x^{1-\alpha} - 1 \over \alpha - 1} = \log(x)$. It is also noted that $\alpha = 0$ corresponds to the sum rate maximization $g_\alpha(x) = x-1$, $\alpha = 1$ corresponds to proportional fairness $g_{\alpha}(x) = \log(x)$, and $\alpha \to +\infty$ corresponds to max-min fairness.

%% file: FairScheduling.tex
In this section we study scheduling algorithms for our setting, where, for each channel realization, a group of users are selected for transmission, with the goal of maximizing some utility function of the long term user rates.

\subsection{Optimal policy}

The optimal policy $\pi^\star$ depends only on the channel gain statistics, $\gamma_1,...,\gamma_K$, however it is usually impractical to compute it due to the difficulty to maximize over $\pi \in \Pi$. A practical approach is to use an iterative scheme. Assume that time is slotted, where $\hv(t) = (h_1(t),...,h_K(t))$ is the vector of channel gains at time $t$. Consider the iterative algorithm which at time slot $t$ selects the group:
\begin{equation}\label{eq:scheduling_gradient}
	\Jc(\vh(t),t) \in \argmax_{\Jc\subseteq[K]} \left\{ {1 \over T(m,|\Jc|)} \log(1 + \min_{j \in \Jc} h_j(t)) \sum_{i=1}^K {\indic\{i \in \Jc \} \over u_i(t)^{\alpha}} \right\},
\end{equation}
where $\vu(t) = (u_1(t),...,u_K(t))$ is the vector of empirical data rates up to time $t$, and obeys the recursive equation:
\begin{align*}
	 u_i(t+1) = {1 \over t+1}\left[ t u_i(t) +  {\indic\{i \in \Jc(\vh(t),t) \} \over T(m,|\Jc(\vh(t),t))|)} \log(1 + \min_{j \in \Jc(\vh,t)} h_j(t))\right].
\end{align*}
\begin{proposition}\label{prop:gradient}
	Under the above scheme $\vu(t)$ converges almost surely to a utility optimal allocation:
	$$
		{1 \over K} \sum_{i=1}^K g_\alpha(u_i(t)) \to_{t \to \infty}^{a.s.}  \underset{\pi \in \Pi}{\max} \left\{ {1 \over K}
                \sum_{i=1}^K g_\alpha(U_i^\pi) \right\}. 
	$$
\end{proposition}
\bp The above scheme is an example of a general class of schemes called \emph{gradient scheduling} schemes. The above result follows from a straightforward application of the results of~\cite{stolyar} which proves the asymptotic optimality of gradient scheduling schemes. \ep

Therefore, utility-optimal scheduling can be achieved simply by applying the above scheme during a large number of time slots. By corollary, we deduce an alternative characterization of the optimal policy which is essential to prove our main result.
\begin{cor}\label{cor:gradient}
	The following scheme yields a utility optimal scheduling:
	$$
		\Jc^\star(\vh) \in \argmax_{\Jc} \left\{ {1 \over T(m,|\Jc|)} \log(1 + \min_{j \in \Jc} h_j) \sum_{i=1}^K {\indic\{i \in \Jc \} \over (U_i^{\pi^\star})^{\alpha}} \right\}. 
	$$
\end{cor}
\bp 
The result holds as a consequence of proposition \ref{prop:gradient}, by letting $t \to \infty$ in~\eqref{eq:scheduling_gradient}. Equation \eqref{eq:scheduling_gradient} indeed defines which group is selected by the above iterative scheme as $t \to \infty$. \ep

\subsection{Threshold policies and complexity}
We also introduce a sub-class of policies called threshold policies. We say that policy $\pi \in \Pi$ is a threshold policy with threshold $c$ if, for any channel realization $\vh$ it selects all users with a channel gain larger than $c$, that is:
$$
	{\Jc}^\pi(\vh) = \{i=1,\dots,K: h_i \ge c \}.
$$
While threshold policies are in general suboptimal, they can be implemented with minimal complexity. Indeed, computing the solution of \eqref{eq:scheduling_gradient}  can be done in time ${\cal O}(K^2)$, by sorting $\vh$ and searching over the possible values of $|\Jc|$ and $\min_{j \in \Jc} h_j$ (see appendix \ref{subsection:compx} for more details). On the other hand, computing a threshold policy requires ${\cal O}(K)$ time. Furthermore, while computing \eqref{eq:scheduling_gradient} requires all users to report the value of their channel gain $h_1(t),...,h_K(t)$ up to a given accuracy, implementing a threshold policy simply requires user to report $1$ bit of information which is $\indic\{ h_i(t) \ge c \}$.

Surprisingly, as stated in Theorem \ref{th:mainres} of next section, a well designed threshold policy in fact become optimal when the number of users $K$ grows large, so that utility optimal scheduling can be achieved with both linear complexity ${\cal O}(K)$ and 1-bit feedback. 

%% file: SchedulingAnalysis.tex
In this section, we consider utility optimal scheduling when the number of users $K$ grows large. We show that threshold policies become optimal in this regime. 
Our result is general and applies to any value of $\alpha \ge 0$ as well as heterogeneous users where the channel gains statistics $\gamma_1,...,\gamma_K$ are arbitrary as long as they are bounded. We denote by $\underline{\gamma} = \min_{i} \gamma_i$ and $\overline{\gamma} = \max_i \gamma_i$. As a corollary, we compute the optimal threshold policy in closed form as a function of $\gamma_1,...,\gamma_K$, so that the system is indeed tractable.

\subsection{Main result}

We first state Theorem~\ref{th:mainres}, the main technical contribution of this work. That is, as the number of users grows large ($K \to \infty$), a well designed threshold policy become utility optimal, and that the optimal threshold may be derived explicitly as a function of the channel gains statistics $\gamma_1,...,\gamma_K$.

\begin{theorem}\label{th:mainres}
	Consider the solution of the optimization problem:
\begin{equation}\label{eq:optimalthreshold}
	c^\star \in \underset{c \ge 0}{\argmax} \left\{ {1 \over K} \sum_{i=1}^K g_{\alpha}\left(\log(1 + c) e^{-{c \over \gamma_i}}\right) \right\},
\end{equation}
and $\pi$ the threshold policy with threshold $c^\star$. Then the long term data rates under $\pi$ are:
$$
	U_i^\pi = {1 \over T(m,\infty)} \log(1 + c^\star) e^{-{c^\star \over \gamma_i}} + o(1) \,\, , \,\, K \to \infty.
$$

Furthermore, $\pi$ is asymptotically optimal, in the sense that:
$$
	{1 \over K} \sum_{i=1}^K  g_\alpha(U_i^\pi) =  \underset{\pi \in \Pi}{\max} \left\{{1 \over K}  \sum_{i=1}^K  g_\alpha(U_i^{\pi}) \right\} + o(1)  \,\, , \,\, K \to \infty.
$$
\end{theorem}
The proof of theorem~\ref{th:mainres} is long and technical, and is fully detailed in the next subsections. A summary of the proof technique is found in subsection~\ref{subsec:proof_summary}.

\subsection{Optimal threshold}
We now show that, for $\alpha \ge 1$ the optimal threshold defined in~\eqref{eq:optimalthreshold} reduces to the maximization of a concave function, so that it can be computed efficiently using a local search method such as Newton's method.
\begin{proposition}\label{prop:optimalthreshold}
	Consider $c^\star$ the optimal threshold as defined in~\eqref{eq:optimalthreshold}. For $\alpha = 1$, the optimal threshold is given by:
	$$
		c^\star = e^{W_0\left(  K (\sum_{i=1}^K (1/\gamma_i))^{-1} \right)} - 1,
	$$ 
	with $W_0$ the Lambert $W$ function. For $\alpha \ge 1$, the optimal threshold is the unique solution to the equation:
		$$
			(1+c)\log(1+c) = {\sum_{i=1}^K  e^{-{c(1-\alpha) \over \gamma_i}} \over \sum_{i=1}^K  {1 \over \gamma_i} e^{-{c(1-\alpha) \over \gamma_i}}}.
		$$
\end{proposition}
\bp
	In all cases, it is noted that $0 < c^\star < \infty$. Consider $\alpha = 1$. By definition, since $g_{\alpha}(x) = \log(x)$:
	\begin{align*}
		c^\star &\in \underset{c \ge 0}{\argmax} \left\{ {1 \over K} \sum_{i=1}^K g_{\alpha}(\log(1 + c) e^{-{c \over \gamma_i}}) \right\}, \\
			&=  \underset{c \ge 0}{\argmax} \left\{ \log\log(1 + c) - {c \over K} \sum_{i=1}^K {1 \over \gamma_i} \right\}.
	\end{align*}
	Since $c \mapsto \log\log(1 + c)$ is strictly concave, mapping $c \mapsto \log\log(1 + c) - {c \over K} \sum_{i=1}^K {1 \over \gamma_i}$ is strictly concave, hence it admits a unique local maximum which is $c^\star$. The optimal threshold $c^\star$ is thus the unique point at which the derivative is null. Differentiating we get:
	$$
		(1+c^\star)\log(1+c^\star) = K \left(\sum_{i=1}^K {1 \over \gamma_i} \right)^{-1}.
	$$
	The result follows by definition of the Lambert function $W_0$.
	
	Now consider $\alpha > 1$, so that $1 - \alpha < 0$. By definition, since $g_{\alpha}(x) = {x^{1-\alpha} - 1 \over 1 - \alpha}$:
	\begin{align*}
		c^\star &\in \underset{c \ge 0}{\argmax} \left\{ {1 \over K} \sum_{i=1}^K g_{\alpha}(\log(1 + c) e^{-{c \over \gamma_i}}) \right\}, \\
			&= \underset{c \ge 0}{\argmin} \left\{ \sum_{i=1}^K \log(1 + c)^{1-\alpha} e^{-{c(1-\alpha) \over \gamma_i}} \right\}, \\
		&= \underset{c \ge 0}{\argmin} \left\{ (1-\alpha)\log\log(1 + c) + \log\left(\sum_{i=1}^K  e^{-{c(1-\alpha) \over
                \gamma_i}} \right) \right\},
		\end{align*}
		where we took the logarithm to obtain the last expression. Now, since $\alpha > 1$, $c \mapsto (1-\alpha)\log\log(1 + c)$ is convex, and so is $c \mapsto \log\left(\sum_{i=1}^K  e^{-{c(1-\alpha) \over \gamma_i}} \right)$ (log-sum-exp function, see \cite{boyd2004convex}). Hence the above admits a single local minimum, which equals $c^\star$ and may be found by solving:
		$$
			(1+c)\log(1+c) = {\sum_{i=1}^K  e^{-{c(1-\alpha) \over \gamma_i}} \over \sum_{i=1}^K  {1 \over \gamma_i} e^{-{c(1-\alpha) \over \gamma_i}}}.
		$$
\ep

\subsection{Proof element 1: lower bound on the rates}

The first step towards proving Theorem~\ref{th:mainres} is to show that the rates allocated by $\alpha$-fair scheduling are upper and lower bounded by two constants, so that $\min_i 1/(U_i^{\pi^\star})^{\alpha}$ and $\max_i 1/(U_i^{\pi^\star})^{\alpha}$ are of the same order even as $K \to \infty$. This is in fact the step of the proof which is the most involved.

\begin{proposition}
	 There exists $0 < C_1(\underline{\gamma},\overline{\gamma}) < C_2(\underline{\gamma},\overline{\gamma}) < \infty$ such that for all $K \ge 0$ and all $i=1,\dots,K$:
$$
	 C_1(\underline{\gamma},\overline{\gamma}) \le U_i^{\pi^\star} \le C_2(\underline{\gamma},\overline{\gamma}).
$$ 
\end{proposition}
\bp Without loss of generality, we may order users to ensure $U_1^{\pi^\star} \le ... \le U_K^{\pi^\star}$. Throughout the proof we consider the optimal policy $\pi^\star$ and, to ease notation, we denote $U_i^{\pi^\star}$ by $U_i$. We define the function:
\begin{align}\label{eq:f}
f(\Jc,\vh) = {1 \over T(m,|\Jc|)} \log(1 + \min_{j \in \Jc} h_j) \sum_{i=1}^K {\indic\{i \in \Jc \} \over (U_i)^{\alpha}}.
\end{align}
As shown in corollary~\ref{cor:gradient}, under the optimal policy $\pi^\star$, the chosen group is
$$
	\Jc^\star(\vh) \in \arg \max_{\Jc \subset \{1,...,K\}} f(\Jc,\vh).
$$
As a first step, we control the chosen group $\Jc^\star$, in an alternative system when user $1$ is ignored. We define  
$
	\Jc_1^\star(\vh) \in \arg \max_{\Jc \subset \{2,...,K\}} f(\Jc,\vh),
$
the maximizer of $f$ if user $1$ is ignored. Denote by  
$
\bar{U} = \sum_{i=2}^K {1 \over (U_i)^{\alpha}},
$
the sum of weights of all users except user $1$. Define 
$
z = {\bar{U} \over 2 T(m,\infty)} \log(1 +  \underline{\gamma} \log 2).
$
  We now prove the following inequality:
$$
	\PP\left( f(\Jc_1^\star(\vh),\vh) \le z \right) \le {1 \over 2}.
$$
Define the group:
$
	\Jc_1(\vh) = \{i \ge 2: h_i \ge \gamma_i \log 2 \}.
$ 
Let us lower bound $f(\Jc_1(\vh),\vh)$. By definition, $j \in  \Jc_1(\vh)$ implies $h_j \ge \gamma_j \log 2 \ge  \underline{\gamma} \log 2$, hence:
$$ 
\log(1 +  \underline{\gamma} \log 2) \le \log(1 + \min_{j \in \Jc_1(\vh)} h_j),
$$
and further using Property \ref{pt1} implying $T(m,\infty) > T(m,\Jc_1(\vh))$, we obtain the lower bound:
$$
	{1 \over T(m,\infty)} \log(1 +  \underline{\gamma} \log 2) \sum_{i=2}^K {\indic\{h_i \ge \gamma_i \log 2\} \over (U_i)^\alpha} \le f(\Jc_1(\vh),\vh).
$$
Define the random variable:
$$
	Z = {1 \over T(m,\infty)} \log(1 +  \underline{\gamma} \log 2) \sum_{i=2}^K {\indic\{h_i \ge \gamma_i \log 2\} \over (U_i)^\alpha}.
$$
By definition of $\Jc_1^\star(\vh)$, we have $f(\Jc_1^\star(\vh),\vh) \ge f(\Jc_1(\vh),\vh)$, so that:
\begin{align*}
	Z \le f(\Jc_1(\vh),\vh) \le f(\Jc^\star_1(\vh),\vh).
\end{align*}
Since $h_i$ follows an exponential distribution with mean $\gamma_i$, we have $\PP(h_i \ge \gamma_i \log 2) = {1 \over 2}$ and since the channel realizations are independent across users, the random variables $\indic\{h_i \ge \gamma_i \log 2\}$ and $\indic\{h_{i'} \ge \gamma_{i'} \log 2\}$ are independent whenever $i \ne i'$. Therefore:
$$
	\EE(Z) = {1 \over T(m,\infty)} \log(1 +  \underline{\gamma} \log 2) \sum_{i=2}^K {\PP(h_i \ge \gamma_i \log 2) \over (U_i)^\alpha} = z,
$$
and $Z$ is a weighted sum of Bernoulli independent random variables with mean ${1 \over 2}$ so that $Z$ is symmetrical, i.e. $Z - z$ has the same distribution as $z - Z$. Therefore: $\PP(Z \le z) = \PP(Z \ge z) = {1 \over 2}$ and:
$$
		\PP\left( f(\Jc_1^\star(\vh),\vh) \le z \right)  \le \PP(Z \le z) = {1 \over 2}. 
$$
We now control the value of $\min_{i \in \Jc_1^\star(\vh)} h_i$. Choose any $c_1,c_2$ such that both of the conditions below are satisfied:
\begin{align*}
	\text{(i)    } & \log(1 + c_1) < {T(m,1) \over 2 T(m,\infty)} \log(1 +  \underline{\gamma} \log 2); \text{ and } \\
	\text{(ii)    } & {2 T(m,\infty)  \int_{c_2}^\infty (\log(1+y)/\overline{\gamma}) e^{-y/\overline{\gamma}}dy \over T(m,1)
        \log(1 +  \underline{\gamma} \log 2)} \le {1 \over 4}. 
\end{align*}
It is noted that we may indeed choose $c_1,c_2$ in that way since $c \mapsto \log(1 + c)$ is increasing and vanishes for $c=0$, and since $c \mapsto \int_{c}^\infty (\log(1+y)/\overline{\gamma}) e^{-y/\overline{\gamma}}dy$ is decreasing and vanishes for $c \to \infty$. It is also noted that $c_1$ $c_2$ may be chosen only based on the value of $\underline{\gamma}$ and $\overline{\gamma}$ and $m$.

Assume that $\min_{i \in \Jc_1^\star(\vh)} h_i \le c_1$ and that $f(\Jc_1^\star(\vh),\vh) \ge z$. If this event occurs, using the facts that (a) $\log(1 + \min_{i \in \Jc_1^\star(\vh)} h_i) \le \log(1 + c_1)$, and (b) $T(m,|\Jc_1^\star(\vh)|) \ge T(m,1)$ since Property \ref{pt1}, and (c) that $\sum_{i=1}^K {\indic\{i \in \Jc_1^\star(\vh)\}  \over (U_i)^\alpha} \le \bar{U}$, we obtain the upper bound:
$$
	  f(\Jc_1^\star(\vh),\vh) \le {\bar{U} \over T(m,1)}  \log(1 + c_1).
$$
In summary, if $\min_{i \in \Jc_1^\star(\vh)} h_i \le c_1$ and $f(\Jc_1^\star(\vh),\vh) \ge z$ we have $z \le {\bar{U} \over T(m,1)}  \log(1 + c_1)$ and replacing $z$ with its definition:
$$
	{\bar{U} \over 2 T(m,\infty)} \log(1 +  \underline{\gamma} \log 2) \le {\bar{U} \over T(m,1)}  \log(1 + c_1),
$$
which is equivalent to
$$
	 {T(m,1) \over 2 T(m,\infty)} \log(1 + \underline{\gamma} \log 2) \le \log(1 + c_1),
$$
a contradiction with (i) the definition of $c_1$. We have hence proven that $f(\Jc_1^\star(\vh),\vh) \ge z$ implies $\min_{i \in \Jc_1^\star(\vh)} h_i \ge c_1$.

Now assume that $\min_{i \in \Jc_1^\star(\vh)} h_i \ge c_2$ and that $f(\Jc_1^\star(\vh),\vh) \ge z$. If this event occurs, using the facts that  
\begin{align*}
\text{(a)} ~~\log(1 + \min_{j \in \Jc_1^\star(\vh)} h_j) \indic\{i \in \Jc_1^\star(\vh)\} &\le \log(1 + h_i) \indic\{i \in \Jc_1^\star(\vh)\} \\
&\le \log(1 + h_i) \indic\{h_i \ge c_2\},
\end{align*}
since $i \in \Jc_1^\star(\vh)$ implies $h_i \ge c_2$, and (b) $T(m,|\Jc_1^\star(\vh)|) \ge T(m,1)$ since Property \ref{pt1}, we obtain the upper bound:
$$
	f(\Jc_1^\star(\vh),\vh) \le {1 \over T(m,1)} \sum_{i \ge 2} {\log(1 + h_i) \indic\{h_i \ge c_2\} \over (U_i)^\alpha} \equiv Z'.
$$
In summary $\min_{i \in \Jc_1^\star(\vh)} h_i \ge c_2$ and $f(\Jc_1^\star(\vh),\vh) \ge z$ implies $z \le f(\Jc_1^\star(\vh),\vh) \le Z'$. Let us upper bound the expectation of $Z'$. Since $h_i$ has exponential distribution with mean $\gamma_i$ we have:
\begin{align*}
\EE(\log(1+h_i) \indic\{h_i \ge c_2\}) &= \int_{c_2}^\infty (\log(1+y)/\gamma_i) e^{-y/\gamma_i}dy \\
&\le \int_{c_2}^\infty (\log(1+y)/\overline{\gamma}) e^{-y/\overline{\gamma}}dy.
\end{align*}
Hence:
$$
	\EE(Z') \le  {\bar{U} \int_{c_2}^\infty (\log(1+y)/\overline{\gamma}) e^{-y/\overline{\gamma}}dy\over T(m,1)}  .
$$
Using Markov's inequality, we get:
\begin{align*}
& \PP\left( Z' \ge  z \right) \le {\EE(Z') \over z} \\
&\le {2 T(m,\infty)  \int_{c_2}^\infty (\log(1+y)/\overline{\gamma}) e^{-y/\overline{\gamma}}dy \over T(m,1) \log(1 +  \underline{\gamma} \log 2)} \\ & \le {1 \over 4},
\end{align*}
using the definition of $c_2$ for the final inequality.

In conclusion, we have proven that:
\begin{align*}
	 \PP\left(  \min_{i \in \Jc_1^\star(\vh)} h_i \not\in [c_1,c_2]\right) &= \PP\left(  \min_{i \in \Jc_1^\star(\vh)} h_i \not\in [c_1,c_2]; f(\Jc_1^\star(\vh),\vh) \le z\right) \\ 
	 &+ \PP\left(  \min_{i \in \Jc_1^\star(\vh)} h_i \not\in [c_1,c_2];f(\Jc_1^\star(\vh),\vh) \ge z \right)\\
	 &\leq \PP\left( f(\Jc_1^\star(\vh),\vh) \le z\right)+\PP\left(  \min_{i \in \Jc_1^\star(\vh)} h_i \not\in [c_1,c_2];f(\Jc_1^\star(\vh),\vh) \ge z \right)\\
	 &= \PP\left( f(\Jc_1^\star(\vh),\vh) \le z\right)+\PP\left(  \min_{i \in \Jc_1^\star(\vh)} h_i \geq c_2;f(\Jc_1^\star(\vh),\vh) \ge z \right)\\
	&\le \PP\left(Z \le z \right) + \PP\left( Z' \ge  z \right) \\
	&\le {1 \over 2}+ {1 \over 4} = {3 \over 4},
\end{align*}
hence:
$$
	\PP\left(  \min_{i \in \Jc_1^\star(\vh)} h_i \in [c_1,c_2]\right) \ge {1 \over 4}.
$$
The second step involves lower bounding $U_1$, using the previous result on the fluctuations of $\min_{i \in \Jc_1^\star(\vh)} h_i$. We will use the four following facts: (a) Since $\Jc_1^\star(\vh)$ depends solely on $h_2,...,h_K$, the event $\min_{i \in \Jc_1^\star(\vh)} h_i \in [c_1,c_2]$ is independent of $h_1$, (b) When both $\min_{i \in \Jc_1^\star(\vh)} h_i \in [c_1,c_2]$, and $h_1 > c_2$, then $1 \in \Jc^\star(\vh)$ since ${1 \over (U_1)^\alpha} \ge \max_{i \ge 2} {1 \over (U_i)^\alpha}$ and $\min_{i \in \Jc^\star(\vh)} h_i \le c_2 \le h_1$. Indeed, if $1 \not\in \Jc^\star(\vh)$, for any $i \in \Jc^\star(\vh)$ we have $f(\Jc^\star(\vh) \setminus \{i\} \cup \{1\},\vh) > f(\Jc^\star(\vh),\vh)$, a contradiction since $\Jc^\star(\vh)$ is a maximizer of $\Jc \mapsto f(\Jc,\vh)$, (c) Since $h_1$ has exponential distribution with mean $\gamma_1 \ge \underline{\gamma}$, $\PP(h_1 \ge c_2) = e^{-c_2/\gamma_1} \ge e^{-c_2/\underline{\gamma}}$ and (d) We have $T(m,|\Jc^\star(\vh)|) \le T(m,\infty)$ since Property \ref{pt1}.

Putting (a), (b), (c) and (d) together we get:
\begin{align*}
	U_1 &\ge {1 \over T(m,\infty)}\log(1 + c_1) \PP( \min_{i \in \Jc_1^\star(\vh)} h_i \in [c_1,c_2], h_1 \ge c_2) \\
	&= {1 \over T(m,\infty)} \log(1 + c_1) \PP( \min_{i \in \Jc_1^\star(\vh)} h_i \in [c_1,c_2])\PP(h_1 \ge c_2) \\
	&\ge {1 \over T(m,\infty)} {1 \over 4} \log(1 + c_1) e^{-c_2/\underline{\gamma}} \equiv C_1(\underline{\gamma},\overline{\gamma}).
\end{align*}
Furthermore, for any $i=1,...,K$:
\begin{align*}
	U_i &\le {1 \over T(m,1)} \EE(\log(1 + h_i)) \\
	&\le  {1 \over T(m,1)}\log(1 + \EE(h_i)) \\
	&= {1 \over T(m,1)} \log(1 + \gamma_i) \\
	&\le {1 \over T(m,1)} \log(1 + \overline{\gamma}) \equiv C_2(\underline{\gamma},\overline{\gamma}).
\end{align*}
We have proven that:
$$
	C_1(\underline{\gamma},\overline{\gamma}) \le U_i \le C_2(\underline{\gamma},\overline{\gamma})
$$
for all $i=1,...,K$ and all $K$ as announced.
\ep

\subsection{Proof element 2: asymptotic size of $\Jc$}

	From the first proof element we deduce the second one, that is, only groups $\Jc^\star(\vh)$ of large size are chosen with high probability as the number of users grows. In turn this implies that $T(m,|\Jc^\star(\vh)|) \PtoK T(m,\infty)$. This result is important, since it allows to take $T(m,|\Jc^\star(\vh)|)$ out of the equation when it comes to controlling which users are selected by the optimal policy.
\begin{proposition}
		For all $J \ge 0$ we have:
$$
	\PP(|\Jc^\star(\vh)| \ge J) \toK 1.
$$
	Furthermore, $T(m,|\Jc^\star(\vh)|) \PtoK T(m,\infty)$.
\end{proposition}
\bp
	Consider the following group of users:
$$
	\Jc(\vh) = \{i \ge 1: h_i \ge \gamma_i \log 2 \}.
$$ 
	Let us lower bound the value of $f(\Jc(\vh),\vh) = {1 \over T(m,|\Jc|)} \log(1 + \min_{j \in \Jc} h_j) \sum_{i=1}^K {\indic\{i \in \Jc \} \over (U_i)^{\alpha}}$ as defined in \eqref{eq:f}. Using the facts that (a) $T(m,\Jc(\vh)) \le T(m,\infty)$ due to Property \ref{pt1}, (b)  $i \in \Jc(\vh)$ implies $h_i \ge \gamma_i \log 2 \ge \underline{\gamma} \log 2$ so that $\min_{i \in \Jc(\vh)} h_i \ge \underline{\gamma} \log 2$ and (c) $U_i \le C_2(\underline{\gamma},\overline{\gamma})$ so that ${1 \over (U_i)^\alpha} \ge {1 \over C_2(\underline{\gamma},\overline{\gamma})^\alpha}$ we obtain the lower bound:
	$$
		{\log(1 + \underline{\gamma} \log 2) \over  C_2(\underline{\gamma},\overline{\gamma})^\alpha T(m,\infty)} \sum_{i=1}^K \indic\{ h_i \ge \gamma_i \log 2 \} \le f(\Jc(\vh),\vh).
	$$
	Let us upper bound the value of $f(\Jc^\star(\vh),\vh)$, using the facts that (a) $T(m,\Jc(\vh)) \ge T(m,1)$ due to Property \ref{pt2},  (b)  $U_i \ge C_1(\underline{\gamma},\overline{\gamma})$ so that ${1 \over (U_i)^\alpha} \le {1 \over C_1(\underline{\gamma},\overline{\gamma})^\alpha}$ and (c) $\min_{i \in \Jc^\star(\vh)} h_i \le \max_{i=1,...,K} h_i \le \overline{\gamma} \max_{i=1,...,K} (h_i/\gamma_i)$ we obtain:
$$
	f(\Jc^\star(\vh),\vh) \le |\Jc^\star(\vh)| {\log(1 +  \overline{\gamma} \max_{i=1,...,K} (h_i/\gamma_i)) \over  C_1(\underline{\gamma},\overline{\gamma})^\alpha T(m,1)}.
$$	
Since $\Jc^\star(\vh)$ is a maximizer of $\Jc \mapsto f(\Jc,\vh)$ we have $ f(\Jc(\vh),\vh) \le  f(\Jc^\star(\vh),\vh)$, and the two previous inequalities imply:
\begin{align*}
\log(1 + \underline{\gamma} \log 2) & {T(m,1) \over T(m,\infty)} \left({C_1(\underline{\gamma},\overline{\gamma}) \over C_2(\underline{\gamma},\overline{\gamma})}\right)^\alpha  {\sum_{i=1}^K \indic\{ h_i \ge \gamma_i \log 2 \} \over \log(1 +  \overline{\gamma} \max_{i=1,...,K} (h_i/\gamma_i))} \le |\Jc^\star(\vh)|.
\end{align*}
To finish the proof, we prove that:
$$
{\sum_{i=1}^K \indic\{ h_i \ge \gamma_i \log 2 \} \over \log(1 +  \overline{\gamma} \max_{i=1,...,K} (h_i/\gamma_i))}  \astoK \infty.
$$
Since $h_1/\gamma_1,...,h_K/\gamma_K$ are i.i.d exponentially distributed with mean $1$, we have $\PP( h_i \ge \gamma_i \log 2) = {1 \over 2}$ and the law of large numbers gives:
$$
	{1 \over K} \sum_{i=1}^K \indic\{ h_i \ge \gamma_i \log 2 \} \PtoK {1 \over 2}.
$$
Since $\frac{1}{4}<\frac{1}{2}$, we have for $K \to \infty$, with high probability,
$$\sum_{i=1}^K \indic\{ h_i \ge \gamma_i \log 2 \} \ge {K \over 4}.$$
Furthermore,
\begin{align*}
	\PP(\max_{i=1,...,K} (h_i/\gamma_i) \ge 2 \log K) &= 1 - \PP(\max_{i=1,...,K} (h_i/\gamma_i) \le 2 \log K) \\
		&= 1 - \prod_{i=1}^K  \PP( h_i/\gamma_i \le 2 \log K) \\
		&= 1 - \left(1 - {1 \over K^2}\right)^K \toK 0.
\end{align*}
Thus for $K \to \infty$, with high probability, we have
$$\max_{i=1,...,K} (h_i/\gamma_i) \le 2 \log K.$$

Hence, the following occurs with high probability:
\begin{align*}
\log(1 + \underline{\gamma} \log 2){T(m,1) \over T(m,\infty)} \left({C_1(\underline{\gamma},\overline{\gamma}) \over C_2(\underline{\gamma},\overline{\gamma})}\right)^\alpha  
{K \over  \log(1 +  2\overline{\gamma} \log K  )} \le |\Jc^\star(\vh)|.
\end{align*}
Since ${K \over \log \log K} \toK \infty$, this implies that, for all $J \ge 0$:
$$
	\PP(|\Jc^\star(\vh)| \ge J) \toK 1.
$$
Therefore, for any $J \ge 0$:
$$
	\PP\left(T(m,J) \le T(m,|\Jc^\star(\vh)|) \le T(m,\infty)\right) \toK 1.
$$
This holds for all $J$, which proves the second statement.
\ep
\subsection{Proof element 3: convergence to a deterministic equivalent}

The last proof element is to show that, when $K \to \infty$, maximizing $f(\Jc,\vh)$ reduces to a simpler, deterministic optimization problem, which we call a ``deterministic equivalent'' of the original problem. Define the following mapping:
$$
	\phi(\Jc,\vh) = \log(1 + \min_{j \in \Jc} h_j) {1 \over K} \sum_{i=1}^K {\indic\{i \in \Jc \} \over (U_i)^{\alpha}},
$$	
which corresponds to the value of ${T(m,\infty) \over K} f(\Jc,\vh) $ when $|\Jc|$ goes to infinity. Further define $\psi$:
$$
	\psi(c,\vh) = \log(1 + c) {1 \over K} \sum_{i=1}^K {\indic\{h_i \ge c \} \over (U_i)^{\alpha}},
$$
which is the value of $\phi$ when selecting only users whose channel realization is larger than $c$. It is noted that when $K\rightarrow \infty$, we have
$$
	\max_{\Jc \subset \{1,...,K\}} \phi(\Jc,\vh) = \max_{c \ge 0} \psi(c,\vh).
$$
Indeed, if $\min_{j \in \Jc} h_j = c$ for some $c$, then all users $i$ such that $h_i \ge c$ should be included in $\Jc$ in order to maximize $\phi(\Jc,\vh)$. Hence maximizing $\phi(\Jc,\vh)$ over all subsets of users $\Jc$ reduces to a simple, one-dimensonnal search over the value of $\min_{j \in \Jc} h_j = c$, that is maximizing $\psi(c,\vh)$ over $c \ge 0$. We are now left to control the value of the random quantity $\max_{c \ge 0} \psi(c,\vh)$, which is not straightforward since its maximizer $\arg\max_{c \ge 0} \psi(c,\vh)$ is typically a random variable as well. For a fixed value of $c$, we define $\Psi(c)$ which is the expected value of $\psi(c,\vh)$:
$$
\Psi(c) = \EE(\psi(c,\vh)) =  \log(1 + c) {1 \over K} \sum_{i=1}^K {e^{-c/\gamma_i}\over (U_i)^{\alpha}}.
$$
We will show that $\Psi$ constitutes a \emph{deterministic equivalent}, in the sense that maximizing $\psi(c,\vh)$ over $c \ge 0$ for a fixed value of $\vh$ yields, asmptotically with high probability, the same outcome as maximizing $\Psi(c)$ over $c \ge 0$. In other words, a concentration phenomenon occurs as the number of users grows large and channel opportunism does yield any gains over choosing all users whose channel realization is above a fixed threshold.
\begin{proposition}\label{prop:deterministic_equivalent}
	We have:
	$$
		 \max_{c \ge 0} \psi(c,\vh) \PtoK \max_{c \ge 0} \Psi(c).
	$$
\end{proposition}
\bp
	We first show that, for any fixed $c$, $\psi(c,\vh)$ is concentrated around $\Psi(c)$ when $K \to \infty$. Since (a) the
        channel realizations $h_1,...,h_K$ are independent across users, and (b) $\var(\indic\{h_i \ge c \}) \le 1$, and (c) $U_i \ge C_1(\underline{\gamma},\overline{\gamma})$ for $i=1,...,K$, we have:
	\begin{align*}
		\var(\psi(c,\vh)) &=  {\log(1 + c)^2 \over K^2} \sum_{i=1}^K {\var(\indic\{h_i \ge c \}) \over (U_i)^{2 \alpha}} \\
		&\le  {\log(1 + c)^2 \over K C_1(\underline{\gamma},\overline{\gamma})^{2\alpha} } \toK 0.
	\end{align*}
	Hence, Chebychev's inequality proves that 
	$$
	\psi(c,\vh) \PtoK \EE(\psi(c,\vh)) = \Psi(c).
	$$
	 We may now lower bound $\max_{c \ge 0} \psi(c,\vh)$ as follows. Consider $\tilde{c} \in \arg\max_{c \ge 0} \Psi(c)$, then we have 
	 $\psi(\tilde{c},\vh) \le \max_{c \ge 0} \psi(c,\vh)$  
	 and since 
	 $
	 \psi(\tilde{c},\vh) \PtoK \Psi(\tilde{c}) = \max_{c \ge 0} \Psi(c),
	 $ 
	 this proves that, for all $\epsilon > 0$:
	$$
		\PP\left( \max_{c \ge 0} \Psi(c) - \epsilon \le \max_{c \ge 0} \psi(c,\vh)  \right) \toK 1.
	$$
We now upper bound $\max_{c \ge 0} \psi(c,\vh)$. We do so by splitting $[0,+\infty)$ into a finite number of intervals and control the behaviour of $c \mapsto \psi(c,\vh)$ in those intervals. Consider $\epsilon > 0$ fixed. Define $\delta > 0$, and $L \ge 0$ such that both of the following conditions are satisfied:
\begin{align*}
\text{(i) }& {1 \over C_1(\underline{\gamma},\overline{\gamma})^{\alpha}} \int_{L\delta}^{\infty}  (\log(1 + y)/\overline{\gamma}) e^{-y/\overline{\gamma}}dy \le {\epsilon \over 2}, \\
\text{(ii)} & {\delta \over C_1(\underline{\gamma},\overline{\gamma})^{\alpha}} \le {\epsilon \over 2}.
\end{align*}
Such a choice is always possible since $\int_{L\delta}^{\infty}  (\log(1 + y)/\overline{\gamma}) e^{-y/\overline{\gamma}}dy$ vanishes for $L\delta \to \infty$. Further define:
$$
	m_\ell = \begin{cases} 
						\max_{c \in [(\ell-1)\delta,\ell\delta]} \psi(c,\vh) & \text{ if } \ell=1,...,L \\
						\max_{c \in [L\delta,+\infty)} \psi(c,\vh) & \text{ if } \ell=L+1. 
			\end{cases}
$$
It is noted that $m_1,...,m_{L+1}$ are random variables and that:
$$
\max_{c \ge 0} \psi(c,\vh) = \max_{\ell=1,...,L+1} m_\ell.
$$
We may now upper bound the value of each $m_\ell$ individually. First consider $c \in [(\ell-1)\delta,\ell\delta]$, then we have:
\begin{align*}
	\psi(c,\vh) &\le  \log(1+\ell\delta){1 \over K} \sum_{i=1}^K {\indic\{h_i \ge (\ell-1)\delta \} \over (U_i)^{\alpha}}. \\
							&= \psi((\ell-1)\delta,\vh) \\
							&+ (\log(1+\ell\delta) - \log(1+(\ell-1)\delta)){1 \over K} \sum_{i=1}^K {\indic\{h_i \ge (\ell-1)\delta \} \over (U_i)^{\alpha}} \\
							&\le \psi((\ell-1)\delta,\vh) + {\delta \over C_1(\underline{\gamma},\overline{\gamma})^{\alpha}}, \\
							&\le \psi((\ell-1)\delta,\vh) + {\epsilon \over 2},
\end{align*}
since $c \mapsto \log(1+c)$ is increasing, $c \mapsto \indic\{h_i \ge c\}$ is decreasing, $\log(1+\ell\delta) \le \log(1+(\ell-1)\delta) + \delta$, and $U_i \ge C_1(\underline{\gamma},\overline{\gamma})$ for $i=1,...,K$. We have proven that:
$$
	m_{\ell} \le\psi((\ell-1)\delta,\vh) + {\epsilon \over 2} \,,\, \ell=1,...,L
$$
and since
$$
	\psi((\ell-1)\delta,\vh) \PtoK \Psi((\ell-1)\delta) \le \max_{c \ge 0} \Psi(c) \,,\, \ell=1,...,L,
$$
we have that:
$$
	\PP( m_\ell \le \max_{c \ge 0} \Psi(c) + \epsilon) \toK 1 \,,\, \ell=1,...,L.
$$
Now consider  $c \in [L\delta,\infty)$. We have the upper bound:
\begin{align*}
	\psi(c,\vh) \le {1 \over K C_1(\underline{\gamma},\overline{\gamma})^{\alpha}} \sum_{i=1}^K \log(1+h_i)\indic\{h_i \ge
        L\delta\} \equiv Y,
\end{align*}
using the fact that $U_i \ge C_1(\underline{\gamma},\overline{\gamma})$ for $i=1,...,K$ and:
\begin{align*}
	\log(1+c)\indic\{h_i \ge c\} &\le \log(1+h_i)\indic\{h_i \ge c\} \\
																&\le \log(1+h_i)\indic\{h_i \ge L\delta\}.
\end{align*}
Hence $m_{L+1} \le Y$, and we control the first and second moment of $Y$ to show that $Y$ is concentrated around its expectation. By definition of $L$ and $\delta$, since $h_i$ has exponential distribution with mean $\gamma_i$:
\begin{align*}
	\EE(Y) &= {1 \over KC_1(\underline{\gamma},\overline{\gamma})^{\alpha}} \sum_{i=1}^K \EE(\log(1+h_i)\indic\{h_i \ge L\delta\}) \\
				 &= {1 \over KC_1(\underline{\gamma},\overline{\gamma})^{\alpha}} \sum_{i=1}^K \int_{L\delta}^{\infty}  (\log(1 + y)/\gamma_i)  e^{-y/\gamma_i}dy, \\
			   &\le  {1 \over C_1(\underline{\gamma},\overline{\gamma})^{\alpha}} \int_{L\delta}^{\infty}  (\log(1 + y)/\overline{\gamma}) e^{-y/\overline{\gamma}}dy \\
			   &\le {\epsilon \over 2},
\end{align*}
	and since $h_1,...,h_K$ are independent:
\begin{align*}
	\var(Y) &= {1 \over K^2 C_1(\underline{\gamma},\overline{\gamma})^{2\alpha}} \sum_{i=1}^K \var(\log(1+h_i)\indic\{h_i \ge L\delta\}) \\
	&\le {1 \over K C_1(\underline{\gamma},\overline{\gamma})^{2\alpha}} \int_{0}^{+\infty} (\log(1 + y)^2/\overline{\gamma}) e^{-y/\overline{\gamma}}dy \toK 0
\end{align*}	
using the fact that for $i=1,..,.K$:
\begin{align*}
	 \var(\log(1+h_i)  \indic\{h_i \ge L\delta\}) &\le \EE(\log(1+h_i)^2\indic\{h_i \ge L\delta\}^2) \\
	 &\le \EE(\log(1+h_i)^2)\\ 
	 &= \int_{0}^{+\infty} (\log(1 + y)^2/\gamma_i) e^{-y/\gamma_i}dy \\
	 &\le \int_{0}^{+\infty} (\log(1 + y)^2/\overline{\gamma}) e^{-y/\overline{\gamma}}dy.
\end{align*}
	Hence Chebychev's inequality shows that $Y \PtoK \EE(Y) \le {\epsilon \over 2}$, from which we deduce:
	$$
			\PP( m_{L+1} \le \epsilon) \toK 1.
	$$
	So combining both cases, we have that:
	$$
		\PP( m_{\ell} \le  \max_{c \ge 0} \Psi(c) + \epsilon) \toK 1 \;,\; \ell=1,..,L+1.
	$$
	We have proven that, for all $\epsilon > 0$:
	$$
	 		\PP( \max_{c \ge 0} \Psi(c) - \epsilon \le \max_{c \ge 0} \psi(c,\vh)  \le \max_{c \ge 0} \Psi(c) + \epsilon ) \toK 1,
	$$
	and $\max_{c \ge 0} \psi(c,\vh) \PtoK \max_{c \ge 0} \Psi(c)$ as announced.	
\ep

\subsection{Putting it all together}\label{subsec:proof_summary}

	We now complete the proof of Theorem~\ref{th:mainres}. From proposition~\ref{prop:deterministic_equivalent}, asymptotically with high probability, utility optimal scheduling can be realized by selecting a threshold policy, where the threshold $c^\star$ is a maximizer of the deterministic mapping $\Phi$. Under a theshold policy with threshold $c^\star$, the rate of user $i$ is given by:
$$
	U_i = \EE\left( {1 \over T(m,J(c^\star))} \log(1+c^\star) \indic\{ h_i \ge c^\star\}\right), 
$$ 
where $J(c^\star)$ is the number of users whose channel realization is above $c^\star$:
$$
	J(c^\star) = \sum_{i=1}^K \indic\{h_i \ge c^\star\}.
$$
	We have:
	\begin{align*}
		\left|U_i - \EE\left( {1 \over T(m,\infty)} \log(1+c^\star) \indic\{ h_i \ge c^\star\}\right)\right| 
		\le \log(1+c^\star)\EE\left( \left| {1 \over T(m,\infty)} -  {1 \over T(m,J(c^\star))}\right|\right).
	\end{align*}
	From the law of large numbers:
	$$
		{J(c^\star) \over K} \ge {1 \over K} \sum_{i=1}^K \indic\{\underline{\gamma}(h_i/\gamma_i) \ge c^\star\} \astoK e^{-{c^\star \over \underline{\gamma}}} > 0,
	$$
	therefore $J(c^\star) \astoK \infty$. Hence 
	$$
	{1 \over T(m,J(c^\star))} \astoK  {1 \over T(m,\infty)}
	$$
	and 
	$$
	\EE\left({1 \over T(m,J(c^\star))}\right) \le {1 \over T(m,1)} \,\,,\,\, K \ge 1, 
	$$
	so we apply Lebesgue's theorem to yield:
	$$
		\EE\left( \left| {1 \over T(m,\infty)} -  {1 \over T(m,J(c^\star))}\right|\right)  \toK 0.
	$$
	We have proven that:
	\begin{align*}
		U_i &\toK \EE\left( {1 \over T(m,\infty)} \log(1+c^\star) \indic\{ h_i \ge c^\star\}\right) \\
				&= {1 \over T(m,\infty)} \log(1+c^\star) e^{-{c^\star \over \gamma_i}}.
	\end{align*}
	The value of $c^\star$ may be retrieved from the fact that applying theshold policy with threshold $c^\star$ maximizes the utility ${1 \over K} \sum_{i=1}^K g_\alpha(U_i)$, hence:
\begin{align*}
	c^\star &\in \underset{c \ge 0}{\argmax} \left\{ {1 \over K} \sum_{i=1}^K g_{\alpha}\left( {\log(1 + c) e^{-{c \over \gamma_i}} \over T(m,\infty)}\right) \right\}, \\
	&=\underset{c \ge 0}{\argmax} \left\{ {1 \over K} \sum_{i=1}^K g_{\alpha}\left( \log(1 + c) e^{-{c \over \gamma_i}}\right) \right\}.
\end{align*}
This completes the proof of Theorem~\ref{th:mainres}.

%% file: Numerical-Examples.tex

In this section, we illustrate the performance of the various schemes defined in the previous sections through numerical experiments. For each scheme, we compute the long term average data rates of each user $U_1,...,U_K$, and the corresponding utility ${1 \over K} \sum_{i=1}^K g_\alpha(U_i)$, which is our objective function. The considered schemes are recalled below.
\begin{itemize}
\item \underline{Superposition}: At each slot $t$, this scheme solves the weighted sum rate maximization problem in $\Gamma(\hv(t))\subseteq \RR_+^{2^K-1}$, using Theorem \ref{theorem:WSR}:
$$
\Rm_{\rm sp}(\hv(t),t) =\arg\max_{\Rm \in \Gamma(\hv(t))}\sum_{\Jc: \Jc \subseteq [K]} \theta_{\Jc}(t) R_{\Jc} \, \text{ with } \theta_{\Jc}(t)=\frac{\sum_{i\in\Jc}\frac{1}{u_i^{\alpha}(t)}}{T(m,|\Jc|)},
$$
where $u_i(t)$ is the mean empirical rate up to time $t$ for user $i$. The average rate of user $i$ is
$$
U_{{\rm sp},i}= \lim_{t \to \infty} \EE\left[ \sum_{\Jc:i\in\Jc}\frac{1}{T(m,|\Jc|)}R_{{\rm sp},\Jc}(\hv(t),t)\right]. 
$$
\item \underline{Selection with full CSIT}: At each slot $t$, this scheme selects the subset of users 
$$
\Jc_{\rm sc}(\hv(t),t) = \argmax_{\Jc\subseteq[K]} \left\{ {1 \over T(m,|\Jc|)} \log(1 + \min_{j \in \Jc} h_j(t)) \sum_{i=1}^K {\indic\{i \in \Jc \} \over u_i(t)^{\alpha}} \right\}.
$$
The average rate of user $i$ is:
$$
U_{{\rm sc},i}=\lim_{t \to \infty} \EE\left[ \frac{1}{T(m,|\Jc_{\rm sc}(\hv(t),t)|)}\log(1+\min_{ j\in\Jc_{\rm sc}(\hv(t),t)}h_j(t))\indic\{i\in\Jc_{\rm sc}(\hv(t),t)\}\right] .
$$
\item \underline{Threshold-based selection}: At each slot $t$, this scheme selects the subset of users $\Jc_{_{\rm th}}(\hv(t))=\{i: h_i(t)\geq c^\star\}$, where $c^\star$ is the threshold given by \eqref{eq:optimalthreshold}, and depends only on the channel statistics $\gamma_1,...,\gamma_K$. The average rate of user $i$ is:
$$
U_{{\rm th},i}=\EE\left[ \frac{1}{T(m,|\Jc_{{\rm th}}(\hv(t))|)}\log(1+\min_{j\in\Jc_{\rm th}(\hv(t))}h_j(t))\indic\{i\in\Jc_{\rm th}(\hv(t))\}\right] .
$$
\item \underline{Baseline}: At each slot $t$, this scheme selects the subset of users $\{1,...,K\}$, and the average rate of user $i$ is:
$$
U_{{\rm bl},i}=\frac{1}{T(m,K)}\EE\left[\log(1+\min_{1\leq j\leq K}h_j(t))\right]. 
$$
\end{itemize}
In all scenarios, we divide users into two classes of $K/2$ users each: strong users with $\gamma_k = P$ and weak users with $\gamma_k = 0.2P$. For each figure we consider a normalized cache size of $m=[0.1,0.6]$. In Figs. \ref{fig:ka0}, \ref{fig:ka1} and \ref{fig:ka2} we plot the utility versus $K$ for $\alpha=0$, $\alpha=1$ and $\alpha=2$ respectively at $P=10$ dB. In Figs. \ref{fig:pa0}, \ref{fig:pa1} and \ref{fig:pa2} we plot the utility versus $P$ for $\alpha=0$, $\alpha=1$ and $\alpha=2$ respectively with $K=20$ users. We draw the following conclusions:
\begin{itemize}
\item Complexity: As seen in Figs. \ref{fig:ka0} and \ref{fig:pa0}, superposition encoding outperforms all the others schemes at the price of a very high complexity of ${\cal O}(2^K)$ compared to other schemes whose complexity is polynomial: ${\cal O}(K^2)$ for the selection scheme with full CSIT and ${\cal O}(K)$ for the threshold-based selection scheme.
\item Number of users $K$: From Figs. \ref{fig:ka0}, \ref{fig:ka1} and \ref{fig:ka2}, the performance of the threshold-based scheme is as good as full CSIT selection scheme for a sufficiently large $K$, as predicted by Theorem~\ref{th:mainres}. In Fig. \ref{fig:ka0}, corresponding to $\alpha=0$, the average per user rate of the baseline scheme vanishes as the number of users increases for both small and large cache size as predicted by Proposition \ref{BaselineScheme}. For $\alpha=1$ and $\alpha=2$, the utility of the baseline scheme decreases with the number of users. On the contrary, the utility of all the other schemes converges to a constant as $K$ grows for all $\alpha$.
\item Power constraint $P$: We observe in Figs. \ref{fig:pa0}, \ref{fig:pa1} and \ref{fig:pa2} that the performance of full CSIT selection, threshold-based selection and baseline schemes becomes identical for large $P$, which is expected since in that case the multicast rate is not limited by users with small channel gains. Therefore, all users are selected. Note that~\cite[proposition 2]{ghorbel2017opportunistic} proves that the full CSIT selection scheme coincides with the baseline scheme in the large $P$ regime for $\alpha=0$.
\item Memory size $m$: Figs. \ref{fig:ka0}-\ref{fig:pa2} show that the gap between the threshold-based scheme and the full CSIT scheme decreases with the memory size. Such a behavior is justified by Property \ref{pt2} stating that the function $k\rightarrow T(m,k)$ converges to $\frac{1-m}{m}$ faster as the memory size $m$ increases. 
\item Alpha-fairness $\alpha$: We now consider the performance as a function of the fairness parameter $\alpha$. We notice that the gap between the selection with full CSIT and the threshold-based selection decreases as the parameter $\alpha$ increases. This is because both schemes tend to coincide with the baseline scheme, or max-min scheduler as $\alpha \to \infty$.
\end{itemize}  
  In summary, remarkably, even for a relatively reasonable number of users, say $K \ge 50$, the threshold-based selection scheme ensures near optimal performance, with both $1$-bit feedback and linear complexity ${\cal O}(K)$, which makes this scheme appealing for practical implementation.

%% file: Appendix.tex
\appendix

\subsection{Proof of Proposition \ref{BaselineScheme}} \label{appendix:prop1}
The content delivery rate is:
\begin{align*}
\Rbl (K,\gamma)= \frac{K}{T(m,K)} \EE( \log(1 + \min_{k=1,...,K} h_k)).
\end{align*}
Since $(h_k)_{k=1,...,K}$ are i.i.d. exponentially distributed with mean $\gamma$, $\min_{k=1,...,K} h_k$ is also exponentially distributed with mean ${\gamma \over K}$. Hence:
\begin{align*}
 	\EE\left[ \log\left(1+ \min_{k=1,...,K} h_k \right)\right] &= \int_{0}^{+\infty} e^{-x} \log\left(1 + {\gamma \over K}x\right) dx \\
 	&= e^{{K \over \gamma}} E_1\left({K \over \gamma}\right),
\end{align*}
which yields statement (i). 

When $K \to \infty$ we have $\frac{K}{T(m,K)} \sim {K m \over 1-m}$ and
$$
\int_{0}^{+\infty} e^{-x} \log\left(1 + {\gamma \over K}x\right) dx \sim {\gamma \over K} \int_{0}^{+\infty} x e^{-x} dx = {\gamma \over K}.
$$
Replacing yields statement (ii). 

When $\gamma \to \infty$, ${K \over \gamma} \to 0$. Since $E_1(x) \sim \log(1/x)$ for $x \to 0$ we obtain statement (iii).

\subsection{Proof of Theorem \ref{theorem:Region}}\label{subsection:th1}

Let $M_{\Kc}$ be the message for all the users in $\Kc\subseteq [K]$ and of size $2^{nR_\Kc}$. 
We first show the converse. 
It follows that the set of $2^K - 1$ independent messages $\{
M_\Kc:\ \Kc\subseteq [K],\, \Kc\ne\emptyset \}$ can be partitioned as 
\begin{align}
\bigcup_{k=1}^K \{ M_\Kc:\ k\in\Kc\subseteq[k]\}. 
\end{align}%
We can now define $K$ independent mega-messages $M'_k := \{ M_\Kc:\ k\in\Kc\subseteq[k]\}$ with rate $R'_k:=\sum_{\Kc:\,k\in\Kc\subseteq[k]}R_\Kc$.
Note that each mega-message~$k$ must be decoded at least by user~$k$ reliably. Thus, the $K$-tuple $(R'_1,\ldots,R'_K)$ must lie inside the
private-message capacity region of the $K$-user BC. Since it is a degraded BC, the capacity region is known~\cite{el2011network}, and we have
\begin{align}
  R'_k &\le \log \frac{1+h_k \sum_{j=1}^k \beta_j}{1+h_k \sum_{j=1}^{k-1} \beta_j}, \quad k=2,\ldots,K, \label{eq:tmp99}
\end{align}%
for some $\beta_j\ge0$ such that $\sum_{j=1}^K \beta_j\le 1$. This establishes the converse.  

To show the achievability, it is enough to use rate-splitting. Specifically, the transmitter first assembles the original messages into
$K$ mega-messages, and then applied the standard $K$-level superposition coding~\cite{el2011network} putting the~$(k-1)$-th signal on top of the~$k$-th signal. The
$k$-th signal has average power $\beta_k$, $k\in[K]$. At the receivers' side, if the rate of the mega-messages are inside the private-message
capacity region of the $K$-user BC, i.e., the $K$-tuple  $(R'_1,\ldots,R'_K)$ satisfies \eqref{eq:tmp99}, then each user $k$ can decode the
mega-message~$k$. Since the channel is degraded, the users~1 to $k-1$ can also decode the mega-message~$k$ and extract its own message.
Specifically, each user~$j$ can obtain $M_{\Jc}$ (if $\Jc\ni j$), from the mega-message~$k$ when $\Jc\subseteq \Kc$ and
$k\in\Kc\subseteq[k]$. This completes the achievability proof.

\subsection{Proof of Theorem \ref{theorem:WSR}}\label{subsection:th2}

The proof builds on the simple structure of the capacity region. We first remark that for a given power allocation of other users, user $k$ sees $2^{k-1}$ messages $\{M_{\Jc}\}$ for $k\in \Jc \subseteq [k]$ with equal channel gain. For a given power allocation $\{\beta_k\}$, the capacity region of these messages is a simple hyperplane characterized by $2^{k-1}$ vertices $C_k \ev_i$ for $i=1, \dots, 2^{k-1}$, where $C_k$ is the sum rate of user $k$ in the RHS of \eqref{eq:9} and $\ev_i$ is a vector with one for the $i$-th entry and zero for the others. Therefore, the weighted sum rate is maximized for user $k$ by selecting the vertex corresponding to the largest weight. This holds for any $k$. 
\subsection{Complexity of selection scheme with full CSIT}\label{subsection:compx}
  Assume that $h_1(t) \ge ... \ge h_K(t)$, i.e. $\vh(t)$ has been previously sorted. Define $k=\max \Jc(\hv(t),t)$ the index of the worst user and the set size $s=|\Jc(\hv(t),t)|$. Let $\nu_{k}$ be a permutation on $\{1,...,k\}$ such that $u_{\nu_{k}(1)}(t) \le ... \le u_{\nu_{k}(k)}(t)$. Since $\Jc(\hv(t),t)$ is a maximizer of~\eqref{eq:scheduling_gradient}:
\begin{align*}
	{\log(1+ h_k(t)) \over T(m,s)}  \sum_{i=1}^K {\indic\{i \in \Jc(\hv(t),t) \} \over u_i(t)^{\alpha}}&={\log(1+ h_k(t)) \over T(m,s)} \max_{\Jc\subseteq[K] : |\Jc|=s,\max\Jc=k} \sum_{i=1}^K  {\indic\{i \in \Jc \} \over u_i(t)^{\alpha}} \\
						&=  {\log(1+ h_k(t)) \over T(m,s)} \sum_{i=1}^{s}\frac{1}{(u_{\nu_{k}(i)}(t))^\alpha}.
\end{align*}
This implies:
\begin{align*}
\Jc(\hv(t),t)=\{\nu_{k}(1),...,\nu_{k}(s)\}.
\end{align*}
Hence $\Jc(\hv(t),t)$ can be computed by sorting $\vh(t)$ and $\vu(t)$, (with complexity ${\cal O}(K \log(K))$ using quick sort), and searching over the possible values of $k=1,...,K$ and $s=1,...,K$ (with complexity ${\cal O}(K^2)$). Thus, finding $\Jc(\hv(t),t)$ takes time ${\cal O}(K^2)$.